\documentclass{llncs}

\usepackage{graphicx}
\usepackage{amsmath}
\usepackage{amssymb}
\usepackage{color}
\usepackage{mathtools}
\usepackage{a4wide}

\newcommand{\snew}[1]{\textcolor{black}{#1}}

\newcommand{\twu}{\textcolor{black}}

\newcommand{\markj}[1]{\textcolor{black}{#1}}

\newcommand{\sfinal}[1]{\textcolor{black}{#1}}

\newtheorem{observation}{Observation}

\begin{document}


\title{Treewidth of display graphs: bounds, brambles and applications}


\author{Remie Janssen\inst{1}, Mark Jones\inst{1}, Steven Kelk\inst{2}, Georgios Stamoulis\inst{2}, Taoyang Wu\inst{3}}

\institute{
Delft Institute for Applied Mathematics, Delft University of Technology, Netherlands. \email{remiejanssen@gmail.com, markelliotlloyd@gmail.com}
\and
Department of Data Science and Knowledge Engineering (DKE),\\ Maastricht University, P.O. Box 616, 6200 MD Maastricht, The Netherlands. \email{steven.kelk@maastrichtuniversity.nl, georgios.stamoulis@maastrichtuniversity.nl}
\and
School of Computing Sciences, University of East Anglia, United Kingdom.\\\email{Taoyang.Wu@uea.ac.uk}
}






\maketitle

\begin{abstract}
Phylogenetic trees and networks are leaf-labelled graphs used to model evolution.  Display graphs are created by identifying common leaf labels in two or more phylogenetic trees or networks.
The treewidth of such graphs is bounded as a function of many common dissimilarity measures between phylogenetic trees and this has been leveraged in fixed parameter tractability results. Here we further elucidate the properties of display graphs and their interaction with treewidth. We show that it is \textbf{NP}-hard to recognize display graphs, but that display graphs of bounded treewidth can be recognized in linear time. Next we show that if a phylogenetic network displays (i.e. topologically embeds) a phylogenetic tree, the treewidth of their display graph is bounded  \markj{by} a function of the treewidth of the original network (and \markj{also by} various other parameters). In fact, using a bramble argument we show that this treewidth bound is sharp up to an additive term of 1. We leverage this bound to give an FPT algorithm, parameterized by treewidth, for determining whether a network displays a tree, which is an intensively-studied problem in the field. We conclude with a discussion on the future use of display graphs and treewidth in phylogenetics.
\end{abstract}


\section{Introduction}\label{se:intro}
A phylogenetic tree on a set of species (or, more abstractly, \emph{taxa}) $X$ is
a tree whose leaves are bijectively labelled by $X$. The central idea of such structures is that internal nodes represent hypothetical ancestors of $X$ \cite{SempleSteel2003}. In this way, the tree can be viewed as a summary of how $X$ evolved over time. Here we focus on unrooted, binary trees: internal nodes all have degree 3, and there is no direction on the edges of the tree. This is not an onerous restriction, since many phylogenetic inference methods construct unrooted, binary trees. We refer the reader to \cite{steel2016phylogeny,felsenstein2004inferring} for further background on phylogenetics.

In this article we study \emph{display graphs}. Simply put, a display graph is obtained from two or more phylogenetic trees by identifying leaves with the same label \cite{bryant2006compatibility,vakati2011graph,kelk2015}. Display graphs have attracted interest in recent years because of the phenomenon that, if two or more phylogenetic trees are (in some formal sense) ``similar", the \emph{treewidth} of their display graph is bounded \snew{by a function of various parameters. For example, by the number of trees that form the display graph \cite{bryant2006compatibility}, or by the Tree Bisection and Reconnect (TBR) distance of two trees \cite{kelk2015,AllenSteel2001}.}

Treewidth is a well-known graph parameter which measures, at least in an algorithmic sense, how far an undirected graph is from being a tree: many \textbf{NP}-hard problems can be solved in polynomial or even linear time on graphs of bounded treewidth \cite{bodlaender1994tourist,BodlaenderK10,BodlaenderK11}. Display graphs thus form a bridge from phylogenetics into algorithmic graph theory. \snew{In particular, the bounds on the treewidth of display graphs
have been exploited to} give fixed parameter tractable algorithms for a number of \textbf{NP}-hard dissimilarity measures on phylogenetic trees \cite{bryant2006compatibility,kelk2015,baste2017efficient,fernandez2018compatibility}. (See \cite{Cygan:2015:PA:2815661} for background on fixed parameter tractability). Display graphs have also turned out to be useful for speeding up the computation of certain ``easy'' parameters on phylogenetic trees \cite{deng2018fast}, and the treewidth of the display graph itself has also been considered as a proxy for phylogenetic dissimilarity \cite{DBLP:journals/tcs/KelkSW18,grigoriev2015low}.

\snew{The purpose of this article is to further investigate, and algorithmically exploit, properties of the display graphs formed not only by trees, but also by trees and \emph{networks}. To the best of our knowledge this is the first time tree-network display graphs have been considered.}
In the first part of the article, we list some basic properties of display graphs, and then address the problem of \emph{recognizing} them, \snew{a problem posed in \cite{DBLP:journals/tcs/KelkSW18}}. Specifically: given a cubic graph $G$, do there exist two unrooted binary phylogenetic trees $T_1, T_2$ on the same set of taxa $X$ such that $G$ is the display graph $D(T_1, T_2)$ of $T_1$ and $T_2$ (after suppression of degree-2 nodes)? We prove that the problem is \textbf{NP}-hard, by providing an equivalence with the \textbf{NP}-hard \textsc{TreeArboricity} problem \cite{chang2004vertex}. On the positive side, we prove that if $G$ has bounded treewidth then this question can be answered in linear time. For this purpose we use Courcelle's Theorem  \cite{Courcelle90,Arnborg91}. This well-known meta-theorem states, essentially, that graph properties which can be expressed as a bounded-length fragment of Monadic Second Order Logic (MSOL) can be solved in linear time on graphs of bounded treewidth. We provide such an expression for recognizing display graphs.

In the second, longer part of the article, we turn our attention to display graphs formed by merging an unrooted binary phylogenetic tree $T$ with an unrooted binary phylogenetic \emph{network} $N$, both on the same set of taxa $X$. The latter is simply an undirected graph where internal nodes have degree 3 and leaves, as usual, are bijectively labelled by $X$. Unlike trees, networks do not need to be acyclic. We emphasize that unrooted phylogenetic networks (as defined here and in e.g. \cite{GBP2012,van2017unrooted,francis2018tree,solis2016inferring}) should be viewed as undirected analogues of rooted phylogenetic networks, which correspond to directed graphs \cite{HusonRuppScornavacca10}. This is to distinguish them from \emph{split} networks which are phylogenetic data-visualisation tools and which have a very different phylogenetic interpretation; these are sometimes also referred to as ``unrooted'' networks \cite{davidbook}.

\snew{Display graphs involving networks} are relevant because of the growing number of optimization problems, traditionally posed on rooted trees and networks, which are now being mapped to the unrooted setting (see e.g. \cite{keijsper2014reconstructing,van2017unrooted,huber2016transforming,francis2018tree}). We prove that, if $N$ \emph{displays} $T$ - i.e. $N$ contains a topological embedding of $T$ - the treewidth of their display graph is at most $2tw(N) + 1$, where $tw(N)$ is the treewidth of the network $N$. We also give alternative  upper bounds for the treewidth of the display graph of $N$ and $T$ expressed in terms of a parameter more familiar to the phylogenetics community. Specifically, we give (tight) bounds in terms of the \textit{level} of the original network $N$ \cite{GBP2012} (which automatically implies bounds in terms of the weaker parameter \textit{reticulation number)}. Briefly, the level of a network $N$ is simply the maximum, ranging over all biconnected components of $N$, of the number of edges in the biconnected component minus the number of edges that a spanning tree for that component has. Following \cite{kelk2015} we use these upper bounds to give a compact MSOL-based fixed-parameter tractable algorithm for the \textbf{NP}-hard problem of determining whether an unrooted network $N$ displays $T$, under various parameterizations. This problem, particularly in the rooted setting, continues to attract significant interest in the phylogenetics literature (see \cite{gunawan2016program,van2017unrooted,ISS2010b} for relevant references). The parameterization in terms of treewidth is potentially interesting since, as we point out, the treewidth of $N$ can be significantly lower than the level or reticulation number of $N$.

The question arises whether the bound $2tw(N) + 1$ can be strengthened. We show that, up to the additive $+1$ term, this bound is essentially sharp. We do this by providing an infinite family of networks $N$ with corresponding trees $T$ such that $T$ is displayed by $N$ and whereby the treewidth of the display graph is at least twice the treewidth of $N$. To derive the lower bound on treewidth we crucially use \textit{brambles} \cite{DBLP:journals/jct/SeymourT93}.

In the final part of the article we reflect on the potential future use of display graphs and treewidth in phylogenetics, and list a number of open problems.
\section{Preliminaries}
\label{sec:preliminaries}

An \textit{unrooted binary phylogenetic tree} $T$ on a set of leaf labels (known as \textit{taxa}) $X$ is an undirected tree where all internal vertices have degree three and
the leaves are bijectively labeled by $X$.
When it is understood from the context we will often drop the prefix ``unrooted binary phylogenetic'' for brevity.
Similarly, an \textit{unrooted binary phylogenetic network} $N$ on a set of leaf labels $X$ is a simple, connected, undirected graph
that has $|X|$ degree-1 vertices that are bijectively labeled by $X$ and any other vertex has degree 3. \sfinal{See Figure \ref{fig:net} for a simple example of a tree $T$ and a network $N$.}

The \textit{reticulation number} $r(N)$ of a network $N = (V,E)$ is
defined as $r(N) := |E| - (|V|-1)$, i.e., the number of edges we need to delete from $N$ in order to obtain a tree that spans $V$. A network $N$ with $r(N) = 0$ is
simply an unrooted phylogenetic tree. Note that in graph theory the value $|E| - (|V|-1)$ of a connected graph is sometimes called the \emph{cyclomatic number} of the graph \cite{diestel2010}.

For a given network $N$ we define its \emph{level}, denoted $\ell(N)$, as the minimum reticulation number ranging over all biconnected components of $N$. To be consistent with the phylogenetics literature we say that $N$ is a ``level-$k$ network'' if $\ell(N) \leq k$ (which means that they are ``almost $k$-trees'' \cite{bodlaender1998partial}). A level-0 phylogenetic network is simply a phylogenetic tree.
Many \textbf{NP}-hard problems in phylogenetics that involve phylogenetic networks as input or output can be solved in polynomial time if the network has bounded level (or bounded reticulation number) \cite{kelk2014constructing,fischer2015computing,bordewich2017fixed}. 

We now formally define the main object of study in this article, namely the \textit{display graph}:

\begin{definition}
Let $T_1 = (V_1 \cup X, E_1), T_2=(V_2 \cup X, E_2)$ be two trees, both on the same set of leaf labels $X$. The \emph{display graph} of $T_1, T_2$, denoted by
$D(T_1,T_2)$, is formed by \emph{identifying} vertices with the same leaf label and forming the disjoint union of these two trees, i.e., $D(T_1,T_2) = (V_1 \cup V_2 \cup X, E_1 \cup E_2)$.
\end{definition}

 Although the more general definition of display graph encountered in the literature allows the display graph to be formed by more than two trees, not necessarily on the same set of taxa (see e.g. \cite{bryant2006compatibility}), here we will focus exclusively on the above, more restricted definition which is enough for our purposes.
We note that, by construction, a display graph is always biconnected.

Note that a display graph is a labeled graph: the set $X$ bijectively labels the degree-2 nodes in the graph. In some parts of the article the labels $X$ and the degree-2 vertices are not important \snew{(because, modulo some trivial exceptions, degree-2 vertices do not impact upon the treewidth of a graph)}, and in such cases we work with \emph{suppressed} display graphs. Such a graph is obtained by erasing the labels $X$ and repeatedly suppressing degree-2 nodes (i.e. if $\{u,v\}$ and $\{v,w\}$ are edges and $v$ has degree-2, deleting $v$ and its two incident edges and introducing the edge $\{u,w\}$). A suppressed display graph is always cubic (when $|X| \geq 3)$. The act of suppressing degree-2 nodes can potentially create multi-edges. It is easy to see that this happens if and only if the two trees contain one or more common \emph{cherries}. A cherry is a size-2 subset of taxa $\{x, y\}$ that have a common parent, and a cherry is common on two trees if it exists in both of them.

The definition of a display graph formed by a tree $T$ and a network $N$, both on $X$, is completely analogous to the definition for two trees, and is denoted as $D(N,T)$.

Let $N$ be a phylogenetic network and $T$ a phylogenetic tree, both on a common taxon set $X$. Then we say that $N$ \textit{displays} $T$ (or $T$ is displayed by $N$) if there exists a \textit{subtree} $N'$ of $N$ that is a
\emph{subdivision} of $T$ i.e., $T$ can be obtained by a series of edge contractions on a subgraph $N'$ of $N$. We say that $N'$ is an \emph{image} of $T$.  We observe that every vertex of $T$ is mapped to a vertex of $N'$, and that edges of $T$
map to paths in $N'$ (perhaps consisting of only a single edge) leading us to the following observation (see also \cite{bryant2006compatibility}):

\begin{observation}
\label{obs:display}
If an unrooted binary phylogenetic network $N$ displays  an unrooted binary phylogenetic tree $T$, both on the same set of leaf labels $X$, then there exists \twu{ a subtree $N'$ of $N$ and a surjective \sfinal{function $f$} from $N'$ to
$T$ such that: }
\begin{enumerate}
\item[(1)] $f(\ell) = \ell, \forall \ell \in X$,
\item[(2)]
 the subsets of $V(N')$ induced by $f^{-1}$ are mutually disjoint, and each such subset induces a connected subtree of $V(N')$,
$\forall v \in V(T)$, the set $\{ u \in V(N'): f(u) = v \}$ forms a connected component in $N$,
and
\item[(3)] $\forall \{ u,v \} \in E(T), \exists_1 \{ \alpha, \beta \} \in E(N'):$ $f(\alpha) = u$ and $f(\beta) = v$.
\end{enumerate}
\end{observation}

This observation will be crucial when we study the treewidth of $D(N,T)$ as a function of several parameters (including the treewidth) of $N$.

We now move on to define the concept of the \emph{treewidth} of an undirected graph:

\begin{definition}
Given an undirected graph $G=(V,E)$, a \emph{tree decomposition} of $G$ is a pair $(\mathcal{B}, \mathbb{T})$ where $\mathcal{B} = \{B_1, \dots ,B_q\}$  is a multiset of \emph{bags}
and $\mathbb{T}$ is a tree whose $q$ nodes are in bijection with $\mathcal{B}$, satisfying the following three properties:

\begin{itemize}
\item[(tw1)] $\cup_{i=1}^q B_i = V(G)$;
\item[(tw2)] $\forall e = \{ u,v \} \in E(G), \exists B_i \in \mathcal{B} \mbox{ s.t. } \{u,v\} \subseteq B_i$;
\item[(tw3)] \emph{running intersection property:} $\forall v \in V(G)$ all the bags $B_i$ that contain $v$ form a connected subtree of $\mathbb{T}$.
\end{itemize}

The \emph{width} of $(\mathcal{B}, \mathbb{T})$ is equal to $\max_{i=1}^q |B_i|-1$. The \emph{treewidth} of $G$, denoted by $tw(G)$, is the smallest width among all possible tree decompositions of $G$.  A tree decomposition $\mathcal{T}$ achieving the smallest possible width for a given graph $G$ is called optimal.
\end{definition}
\sfinal{If an undirected graph $H$ can be obtained from a graph $G$
by deleting vertices and edges and contracting edges, then $H$ is a \emph{minor} of $G$. It is well known that, if $H$ is a \emph{minor} of a graph $G$, then $tw(H) \leq tw(G)$ \cite{diestel2010}.}

In \cite{DBLP:journals/tcs/KelkSW18} it was shown that the treewidth of the display graph of two trees can be, in the worst case, linear in the number of the vertices in the trees. In this article we will explore the relation of the treewidth  of a display graph formed by a phylogenetic network and a tree displayed by that network, and the treewidth (or other parameters) of the network itself.

Finally, we define the \emph{bramble} parameter of a graph, a parameter closely related to treewidth that is very useful when proving lower bounds on treewidth. Given a graph $G$ and two subgraphs $S_1,S_2$ of it, we say that $S_1$ and $S_2$ \emph{touch} if $V(S_1) \cap V(S_2) \neq \emptyset$, \textit{or} some edge of $G$ has one endpoint in $S_1$ and the other in $S_2$. A \emph{bramble} $B$ of $G$ is a set of connected subgraphs of $G$ that pairwise touch. A (sub)set $H \subseteq V(G)$ is a \emph{hitting set} of a bramble $B$ of $G$ if $H$ intersects every element of $B$. The \emph{order} of $B$ is the minimum size of such a hitting set and the bramble number of $G$, denoted by $br(G)$, is the maximum, among all possible brambles, order of a bramble of $G$. The usefulness of brambles comes from the following result, due to Seymour \& Thomas, relating the treewidth of a graph $G$ to its bramble number:

\begin{theorem}[\cite{DBLP:journals/jct/SeymourT93}]
For any graph $G$ we have that $tw(G) = br(G) -1$.
\end{theorem}



\section{Recognizing display graphs of pairs of trees}
\label{sec:recognition}

We consider the \textsc{DisplayGraph} decision problem, posed in \cite{DBLP:journals/tcs/KelkSW18}: 

\medskip
\noindent
\markj{\textbf{Input:} A biconnected, cubic, simple graph $G=(V,E)$.}

\noindent\markj{\textbf{Goal} Find two unrooted binary trees $T_1, T_2$, on the same set of taxa $X$, such that the suppressed display graph $D(T_1,T_2)$ of these two trees is isomorphic to $G$, if they exist.}

\smallskip
\markj{Note that in this formulation we can assume without any loss of generality that $T_1$ and $T_2$ do not have common cherries.}

Here we will argue that the \textsc{DisplayGraph} problem is \textbf{NP}-hard by providing an equivalence between the \textsc{DisplayGraph} problem and the \textbf{NP}-hard \textsc{TreeArboricity} problem \cite{chang2004vertex} which is defined as follows:

\medskip
\noindent
\textbf{Input:} A simple, undirected graph $G=(V,E)$.

\noindent\textbf{Goal} Find the smallest positive integer $k$ such that there exists a partition $(V_1,\ldots, V_k)$ of $V$ such that each part of the partition induces a tree, i.e., $G|_{V_i}$ is a tree for $i\in[k]$ (such a partition is called a \emph{tree partition}). This $k$ is the \emph{Tree Arboricity} of $G$, also denoted as $ta(G)$.

\smallskip
We emphasize that unlike some closely related variants of the problem (for example \textsc{VertexArboricity} \cite{raspaud2008vertex}), it is not permitted that a $G|_{V_i}$ induces a forest consisting of two or more components.

\sfinal{Chang et al. \cite{chang2004vertex} discuss the decision version of the \textsc{TreeArboricity} problem with $k=2$ (i.e. is $ta(G) \leq 2$?). The following lemma binds their problem to ours.}

\begin{lemma}\label{lem:equivalence}
Given a \snew{simple, connected, cubic} graph $G$ \snew{as} input to the \textsc{TreeArboricity} \snew{decision} problem, $G$ is a ``yes" instance for the \textsc{TreeArboricity} problem with $k=2$ if and only if $G$ is a \snew{suppressed} display graph $D(T_1,T_2)$ of two binary phylogenetic trees $T_1,T_2$ on a common set of taxa $X$.
\end{lemma}

\begin{proof}
Given such $T_1, T_2$ then the partition of the set of vertices into two sets $V_1, V_2$ is simply $V_i = V(T_i) \setminus X$. We exclude the taxa $X$ since, when we form the display graph $D(T_1,T_2)$, these will become degree-2 vertices which are subsequently suppressed. On the other hand, given a bipartition $V_1,V_2$ of $G$, we can form the two phylogenetic trees $T_1, T_2$ on a common set of taxa $X$ whose display graph is isomorphic to $G$ as follows. First of all, by definition, $G|_{V_1}, G|_{V_2}$ are trees. Since $G$ is connected and cubic, every leaf vertex $v$ in one bipartition, say $G|_{V_1}$, has exactly 2 neighbor vertices $u_1, u_2$ in $G|_{V_2}$ (i.e., $\{u_1, u_2\} \subseteq V_2$). 
\markj{Subdivide each of the edges $\{v, u_1\}, \{v,u_2\}$ with a new vertex in $X$ (i.e., for $i=1,2$, replace edge $\{v,u_i\}$ with the two edges $\{v,w\}, \{w,u_i\}$, where $w$ is a newly introduced vertex, and include $w \in X$  which is initially empty).}
The points of subdivisions of these ``crossing" edges (having one vertex in each bipartition) are the taxa $X$ of the new trees. Repeat the process on the remaining leaf vertices from $G|_{V_2}$. The same argumentation will also take care of the remaining degree-2 vertices in each of $G|_{V_1}$ and $G|_{V_2}$. To complete the proof, we need to show that the number of the degree-1 plus the degree-2 vertices in $G|_{V_1}, G|_{V_2}$ are equal, such that the two constructed trees are binary phylogenetic trees. Indeed, this \sfinal{will follow} because $G$ is cubic and connected and a ``yes" instance to the \textsc{TreeArboricity} problem. \sfinal{Specifically}, each edge not entirely in $G|_{V_i}$ must have one endpoint in each bipartition. Thus, if we define for every vertex $v \in V_i$ its ``missing" degree in each tree as \sfinal{$\mu(v) = 3 - \mathrm{deg}(v)$} \sfinal{(where here $\mathrm{deg}(v)$ refers to the degree of $v$ in $G|_{V_i}$)}, then we see that $\sum_{v \in V_1} \mu(v) = \sum_{u \in V_2} \mu(u)$ i.e., both constructed trees $T_1,T_2$ are binary and, by construction, on the same set of taxa $X$.
\end{proof}

\begin{theorem}
\textsc{DisplayGraph} is \textbf{NP}-complete.
\end{theorem}

\begin{proof}
The \textsc{DisplayGraph} problem is easily seen to be in \textbf{NP}: a certificate can be the two trees $T_1,T_2$ that form the graph $G$. We only need to check that $D(T_1,T_2)$, after suppressing degree-2 vertices, is isomorphic to $G$, something that can be done in polynomial time since the graph isomorphism problem is polynomially time solvable for graphs of bounded degree \cite{DBLP:journals/jcss/Luks82,DBLP:journals/corr/abs-1802-04659}.
For hardness, \cite{chang2004vertex} prove that the decision version of the \textsc{TreeArboricity} problem with $k=2$ is \textbf{NP}-complete when restricted to a simple, cubic, 3-connected planar graphs.
Thus, let $G$ be a simple, cubic, 3-connected planar graph that is input to the \textsc{TreeArboricity} problem. A 3-connected graph is vacuously also a biconnected graph, so $G$ is a valid input to the \textsc{DisplayGraph} problem. The result follows because of the if and only if relationship described in Lemma \ref{lem:equivalence}.
\end{proof}

\subsection{The fixed parameter tractability of recognizing display graphs of bounded treewidth}

Let $G = (V,E)$ be a simple, biconnected cubic graph. We will use Courcelle's Theorem to test whether $G$ is a suppressed display graph. This will show that the question can be settled in time $O( f(tw(G)) \cdot |V|)$ where $f$ is a function that depends only on the treewidth of $G$. Specifically, when $G$ has bounded treewidth this will yield a linear time algorithm. The constant-length MSOL formulation simply tests whether $ta(G) \leq 2$. \snew{(Clearly, $ta(G) \geq 2$ because $G$ is not acyclic). The MSOL formulation (and an introduction to MSOL proofs) is given in the appendix.}

\begin{theorem}
Suppressed display graphs can be recognized in linear time on graphs of bounded treewidth.
\end{theorem}
\begin{proof}
This is a consequence of the correctness of the MSOL formulation described in Appendix \ref{subsec:msolrecog} and the equivalence stated in Lemma \ref{lem:equivalence}.
\end{proof}

\section{Display graphs formed from trees and networks}
\label{sec:networks}

In this section we will consider the display graph formed by an unrooted binary phylogenetic network $N = (V,E)$ and an unrooted binary phylogenetic tree $T$ both on the same set of taxa $X$. We will show upper and lower bounds on the treewidth of $D(N,T)$ in terms of the treewidth $tw(N)$ of $N$ and the level $\ell(N)$ of $N$ (and thus also the reticulation number $r(N)$ of $N$). We will also show how these upper bounds can be leveraged algorithmically to give FPT results for \twu{ deciding} whether a given network $N$ displays a given tree $T$.

\subsection{\sfinal{Treewidth upper bounds}}

We first relate the treewidth of the display graph with the treewidth of the network $N$.

\begin{lemma}
\label{lem:twbound} Let $N = (V,E)$ be an unrooted binary phylogenetic network and $T$ an unrooted binary phylogenetic tree, both on $X$, where $|X| \geq 3$. If $N$ displays
$T$ then $tw( D(N,T)) \leq 2 tw(N) + 1$.
\end{lemma}

\begin{proof}Since $N$ displays $T$, \twu{we fix a subgraph $N'$ of $N$ that is a subdivision of $T$ and a surjection function $f$ from $N'$ to $T$ as defined in  Observation 1 (in section Preliminaries)}. Informally, $f$ maps taxa to taxa and  degree-3 vertices of $N'$ to the corresponding vertex of $T$. Each degree-2 vertex of $N'$ lies on a path corresponding to an edge $\{u,v\}$ of $T$; such vertices are mapped to $u$ or $v$, depending on how exactly the surjection was constructed.

Now, consider any tree decomposition $t$ of $N$. Let $k$ be the width of the tree decomposition, i.e., the largest bag in the tree decomposition has size $k+1$. We will
construct a new tree decomposition $t'$ for $D(N,T)$ as follows. For each vertex $u' \in V(N')$ we add $f(u')$ to every bag that contains $u'$.
To show that $t'$ is a valid tree decomposition for $D(N,T)$ we will show that it satisfies all the treewidth conditions. Condition (tw1) holds because $f$ is a surjection.

For property (tw2) we need to show that for every edge $e = \{ u,v \} \in E(T)$, there exists some bag $B \in V(t):
\{ u,v \} \subset B$. For this we use the third property of $f$ described in our observation: $\forall \{ u,v \} \in E(T), \exists_1 \{ \alpha, \beta \} \in E(N'):$ $f(\alpha)
= u$ and $f(\beta) = v$. For each $e = \{u,v\} \in E(T)$, let $\{ \alpha, \beta \} \in E(N')$ be
the edge which is mapped through $f$ to e. Since $\{ \alpha, \beta \} \in E(N)$, there must be a bag $B \in V(t)$ that contains both of $\alpha, \beta$. Since $f(\alpha) = u$ and $f(\beta) = v$,
both of $u,v$ will be added into $B$. For the last property (tw3) we need to show that the bags of $t$ where $u \in V(T)$ have been added form a connected component. For this,
we use property (2) of the function $f$: $\forall v \in V(T)$, the set $\{ u \in V(N'): f(u) = v \}$ forms a connected subtree in $N'$. Hence, the set of bags that contain at least one element from
$\{ u \in V(N'): f(u) = v \}$ form a connected subtree in the tree decomposition. These are the bags to which $v$ is
added, ensuring that (tw3) indeed holds for $v$.

We now calculate the width of $t'$: Observe that the size of each bag can at most double. This can happen when every vertex in the bag \sfinal{is in $V(N')$}
and $f(u') \neq f(v')$ for
every two vertices $u', v'$ in the bag. This causes the largest bag after this operation to have size  at most $2(k+1)$. That is, the width of the new decomposition is at most
$2k+1$.
\end{proof}

We move on and deliver a bound of the treewidth of the display graph $D(N,T)$ in terms of the level $\ell(N)$ of $N$. We remind the reader that a network $N$ is a level-$k$ network if the reticulation number of each biconnected component is at most $k$.

\begin{lemma}
\label{lem:lvlbound}
Let $N=(V,E)$ be an unrooted binary phylogenetic network and $T$ an unrooted binary phylogenetic tree, both on $X$, such that $|X|\geq 3$ and $N$ displays $T$. Then $tw(D(N,T)) \leq \ell(N) + 2$ where $\ell(N)$ is the level of $N$.
\end{lemma}

\begin{proof}
Due to the fact that $N$ displays $T$, there is a subgraph $T'$ of $N$ that is a subdivision of $T$.  If $T'$ is a spanning tree of $N$, then keep $T'$ as is. Otherwise, construct a spanning tree $T'$ of $N$ by greedily adding edges to $T'$ until all vertices of $N$ are spanned. At this point, $T'$ contains exactly $|V|-1$ edges and consists of a subdivision of $T$ from which possibly some unlabelled pendant subtrees (i.e. pendant subtrees without taxa) are hanging.

We argue that $D(T',T)$ has treewidth 2, as follows.  First, note that $D(T,T)$ has treewidth 2, because $T$ is trivally compatible with $T$ (and $|X| \geq 3$) \cite{bryant2006compatibility}. Now, $D(T,T)$ can be obtained from $D(T', T)$ by repeatedly deleting unlabelled vertices of degree 1 and suppressing unlabelled degree 2 vertices. These operations cannot increase or decrease the treewidth \cite{DBLP:journals/tcs/KelkSW18}. Hence, $D(\sfinal{T'}, T)$ has treewidth 2.


For the purposes of the present proof we need a tree decomposition of $D(T,T')$ of width 2 with a very particular structure which we now construct explicitly. For each vertex $a' \in V(T')$ we create a singleton bag $\{a'\}$. For each edge $\{a', b'\} \in E(T')$ we insert the bag $\{a',b'\}$ between the two singleton bags $\{a'\}$ and $\{b'\}$. Now, recall that each vertex $a \in V(T)$ has a unique image $a' \in V(T')$. For each vertex $a \in V(T)$, add $a$ to the singleton bag $\{a'\}$. For each edge $\{a, b\} \in E(T)$, consider the vertices $a'$ and $b'$ in $T'$. We distinguish two cases:

\begin{description}
\item{Case 1.} If $\{a', b'\} \in E(T')$, remove the bag $\{a', b'\}$ that lies between bags $\{a, a'\}$ and $\{b, b'\}$ and replace it with the pair of bags $\{a, a', b\}, \{a', b', b\}$.

\item{Case 2.} If $\{a', b'\} \not \in E(T')$, then edge $\{a,b\} \in V(T)$ corresponds to a path $a', v_1, \ldots, v_t, b'$ in \markj{$T'$} where $t \geq 1$ and none of $v_1, \ldots, v_t$ are images of vertices from $T$. In the tree decomposition, this corresponds to the chain of bags $\{a, a'\}, \{a', v_1\}, \{v_1\}, \{v_1, v_2\}, \{v_2\}, \ldots, \{v_t, b'\}, \{b,b'\}$. In this case, we add $a$ to the bag $\{a', v_1\}$, add both $a$ and $b$ to bag $\{v_1\}$, and add just $b$ to all the remaining bags in the chain.
\end{description}

We denote the tree decomposition by $\mathcal{T}$. It is immediate to verify, by construction, that the above tree decomposition is indeed a valid tree decomposition, i.e., it satisfies all the three properties (tw1)-(tw3).

Crucially, the topology of $\mathcal{T}$ is a subdivision of $T'$: each vertex $a' \in V(T')$ corresponds to a unique bag of $\mathcal{T}$, and each edge in $E(T')$ corresponds to a unique chain of bags in $\mathcal{T}$. We leverage this property as follows.

Let $C$ be a non-trivial biconnected component of $N$. (By non-trivial we mean a biconnected component containing more than 2 vertices. We do this to exclude cut edges, which are formally also biconnected components).  \markj{Let $k = \ell(N)$. Then} we have that $|E(C)|-( |V(C)| - 1) \leq k$. Combined with the fact that $T'$ is a spanning tree of $N$, it follows that we can obtain $N$ from $T'$ by adding at most $k$ missing edges to $C$ (and repeating this for other non-trivial biconnected components). Let $M(C)$ be the at most $k$ edges missing from $C$ and let $A(C)$ be a (not necessarily minimum) minimal vertex cover of the edges in $M(C)$; clearly $|A(C)| \leq k$ since in the worst case we can select one distinct vertex per edge. Due to the topological structure of $\mathcal{T}$ the vertices and edges in $C$ map unambiguously into bags and chains of bags in $\mathcal{T}$. We add
all the vertices in $A(C)$ to all these bags. We repeat this for each non-trivial biconnected component of $N$. Due to the fact that $N$ has maximum degree 3, the non-trivial biconnected components of $N$ are vertex-disjoint, and hence the corresponding bags in $\mathcal{T}$ are all disjoint. This means that, after all the non-trivial biconnected components have been processed, each bag will contain at most $k+3$ vertices.

It remains to show that this is indeed a valid tree decomposition for $D(N,T)$. The vertex set of $D(N,T)$ is the same as that of $D(T, T')$ so (tw1) is clearly
satisfied. For each edge $\{x,y\} \in M(C)$, both $x$ and $y$ are inside $C$, so some bag (in the part of $\mathcal{T}$ corresponding to $C)$ contained $x$ and some bag contained $y$. Given that $A(C) \cap \{x,y\} \neq \emptyset$, adding all the vertices in $A(C)$ to all the bags (corresponding to $C$) ensures that some bag contains both $x$ and $y$. Hence, (tw2) is satisfied. Regarding (tw3), observe that each vertex $x \in A(C)$ lies inside $C$, so in $\mathcal{T}$ some bag (in the part of the decomposition corresponding to $C$) already contained $x$. Moreover, all the bags corresponding to $C$ induce a connected subtree of bags. Hence, adding $x$ to all
these bags cannot destroy the running intersection property for $x$. Hence, (tw3) holds.
\end{proof}

The following observation helps to contextualize Lemmas \ref{lem:twbound} and \ref{lem:lvlbound}.

\begin{observation}
\label{obs:chainbound}
Let $N$ be an unrooted binary phylogenetic network. Then $tw(N)-1 \leq \ell(N) \leq r(N)$.
\end{observation}
\begin{proof}
$\ell(N) \leq r(N)$ follows by definition. To see that $tw(N) - 1 \leq \ell(N)$, it is well-known that the treewidth of a graph is equal to the maximum treewidth ranging over all biconnected components in the graph \cite{bodlaender1998partial}. A spanning tree for each biconnected component can be obtained by deleting at most $\ell(N)$ edges, by definition. A tree has treewidth 1, and adding one edge to a graph can increase its treewidth by at most 1 \cite{bodlaender1998partial}. Hence, each biconnected component has treewidth at most 1+$\ell(N)$. (Alternatively, by observing that level-$k$ networks are almost $k$-trees, \cite[Theorem 74]{bodlaender1998partial} can be leveraged).
\end{proof}

The following corollary is therefore immediate.

\begin{corollary}
\label{cor:reticbound} Let $N = (V,E)$ be an unrooted binary phylogenetic network and $T$ an unrooted binary phylogenetic tree, both on $X$, where $|X| \geq 3$. If $N$ displays
$T$ then $tw( D(N,T)) \leq r(N)+2$. \label{lem:twboundRet}
\end{corollary}

Combining the above results yields the following:

\begin{theorem}
\label{thm:twbound}
Let $N$ be an unrooted binary phylogenetic network and $T$ be an unrooted binary phylogenetic tree, both on $X$. Then if $N$ displays $T$,
$$tw( D(N,T)) \leq \min \bigg\{ 2 tw(N) + 1, r(N) + 2, \ell(N)+2 \bigg\}.$$
\end{theorem}

Note that, from the perspective of $r(N)$ and $\ell(N)$ the bounds $\ell(N)+2$ and $r(N)+2$ are sharp, since if $N=T$ then $r(N)=\ell(N)=0$ and $D(N,T)$ has treewidth 2 \cite{bryant2006compatibility}. Curiously, the treewidth bound gives 3 for this same instance: an additive error of 1. In Section \ref{sec:lb} we will further analyse the sharpness of this bound.

We remark that $tw(N)$ can be arbitrarily small compared to $\ell(N)$ (and $r(N)$). For example, the display graph of two copies of the same tree $T$ on $n$ taxa has treewidth 2. Re-introducing taxa to turn the degree-2 vertices into degree-3 vertices, we obtain a biconnected treewidth 2 phylogenetic network $N=(V,E)$ with $3n-4$ vertices
and $5n-6$ edges, so $\ell(N)=r(N)=|E|-(|V|-1) \rightarrow \infty$ as $n \rightarrow \infty$. However, for $N$ with low $\ell(N)$ the bound $\ell(N)+2$ will \sfinal{potentially}
be stronger than $2tw(N)+1$.

The above bounds raise a number interesting points about the phylogenetic interpretation of treewidth. First, consider the case where a binary network $N$ \emph{does not} display a given binary phylogenetic network $T$. As we can see in Figure \ref{fig:net}, there
is a network $N$ and a tree $T$ such that $N$ does not display $T$ and yet the treewidth of their display graph is equal to the treewidth of $N$ which (as can be easily verified) is equal to three. Hence ``does not display'' does not necessarily cause an increase in the treewidth. On the other hand, our results from \cite{DBLP:journals/tcs/KelkSW18} show that for two incompatible unrooted binary phylogenetic trees (vacuously: neither of which displays the other, and both of which have treewidth 1) the treewidth of the display graph can be as large as linear in the size of the trees. The increase in treewidth in this situation is asymptotically maximal. So the relationship between ``does not display'' and treewidth is rather complex. Contrast this with the bounded growth in treewidth articulated in \sfinal{Theorem \ref{thm:twbound}}. Such bounded growth opens the door to algorithmic applications.


\begin{figure}[ht]
\centering
\includegraphics[scale=0.8]{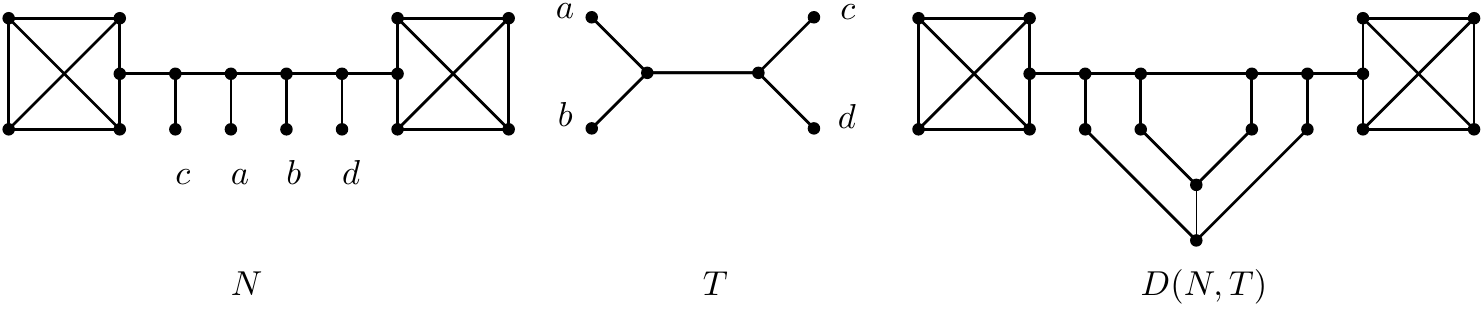}
\caption{The network $N$ does not display the tree $T$ but the treewidth of their display graph is equal to the treewidth of $N$, which is equal to 3. \sfinal{(Note also that, if in $T$ the positions of $b$ and $c$ are swapped, then $N$ does display $T$ but both the network and
the new display graph will still have treewidth 3).}}
\label{fig:net}
\end{figure}

\subsection{An algorithmic application}
We give an example of how the upper bounds from the previous section can be leveraged algorithmically.
The Unrooted Tree Compatibility problem (UTC) is simply the \textbf{NP}-hard problem of determining
whether an unrooted binary phylogenetic network $N = (V,E)$ on $X$ displays an unrooted binary
phylogenetic tree $T$, also on $X$. In \cite{van2017unrooted} a linear kernel is described for the UTC problem and, separately, a bounded-search branching algorithm. Summarizing, these yield FPT algorithms parameterized by $r(N) = |E|-(|V|-1)$
i.e. algorithms that can solve UTC in time at most $f( r(N) ) \cdot \text{poly}( |N| + |T| )$ for some function $f$ that depends only on $r(N)$. \sfinal{We emphasize that these results are more involved than the trivial $2^{r(N)} \cdot \text{poly}( |N| + |T| )$ FPT algorithm for the \emph{rooted} version of the problem.}

Here we give an FPT proof using Courcelle's Theorem. We prove that he problem is FPT when parameterized by $tw(N$). This result has not appeared in the literature before and is potentially interesting given that $tw(N)$ can be much smaller than $\ell(N)$. FPT in terms of $r(N)$ and $\ell(N)$ follow as a corollary of this, due to Observation \ref{obs:chainbound}.

\begin{theorem}
\label{th:twUTC}
Given an unrooted binary phylogenetic network $N = (V,E)$ and an unrooted binary phylogenetic tree both on $X$, we can determine in time $O( f(t) \cdot n )$ whether $N$ displays $T$, where $t$ is $tw(N)$ and $n=|V|$.
\end{theorem}
\begin{proof}
We run Bodlaender's linear-time FPT algorithm \cite{Bodlaender96} to compute a tree decomposition of $D(N,T)$ and return NO if the treewidth is larger than $2t+1$\footnote{The same algorithm can be used to first compute $t$, if it is not known.}. This is
correct by Lemma \ref{lem:lvlbound}. Otherwise, we have a bound on the treewidth of $D(N,T)$ in terms of $t$. Subsequently, we construct the constant-length MSOL sentence described in  Appendix \ref{subsec:whatever} and apply the Arnborg et al. \cite{Arnborg91} variant of
Courcelle's Theorem \cite{Courcelle90}, from which the result follows. (Note that $D(N,T)$ has $O(n)$ vertices and $O(n)$ edges). The result can be made constructive if desired i.e. in the event
of a YES answer the actual set of edge cuts in $N$ (to obtain an image of $T$) can be obtained.
\end{proof}

\begin{corollary}
Given an unrooted binary network $N = (V,E)$ and an unrooted binary tree both on $X$, we can determine in time $O( f(k) \cdot n )$ whether $N$ displays $T$, where $k =
\ell(N)$ and $n=|V|$.
\end{corollary}
\begin{proof}
Immediate from Theorem \ref{th:twUTC} and \twu{Observation}
\ref{obs:chainbound}.
\end{proof}

\subsection{\sfinal{Treewidth lower bounds}}

\label{sec:lb}

In this subsection, we show that the upper bound $tw(D(N,T)) \leq 2tw(N) + 1$ is almost optimal, in the sense that there exist a family of display graphs $D(N,T)$ such that $N$ displays $T$ and $tw(D(N,T)) \geq 2tw(N)$. \sfinal{(Note that, irrespective of whether $N$ displays
$T$, $tw(D(N,T)) \geq tw(N)$ always holds because $N$ is a minor of $D(N,T)$; see Figure \ref{fig:net} for examples when $tw(D(N,T)) = tw(N)$.)}

Fix some integer $r$ and an integer $n$ such that $n > 2r + 2$. We will give a construction for a network $N$ and tree $T$ on a set of $rn$ leaves, such that $tw(N) = r$, $tw(N,T) \geq 2r$, and $N$ displays $T$. For the sake of convenience we will assume that $r$ is even, though the construction can easily be modified to handle cases where $r$ is odd.

The intuition behind the construction is as follows. The network $N$ will have roughly the same structure as an $r \times (n+1)$ grid (with $r$ rows and $n+1$ columns), with leaves attached to the horizontal edges. An $r \times (n+1)$ grid has treewidth $\min(r,n+1) = r$, and so $N$ also has treewidth $r$. The tree $T$ is a long caterpillar that weaves back and forth across the rows of the grid (see Figure~\ref{fig:gridExampleTreeEmbedded}). Thus $T$ is displayed by $N$. However, the display graph $D(N,T)$ has (very roughly) the structure of a $2r \times (n+1)$ grid, and as such can be shown to have treewidth at least $2r$. We remind that a caterpillar graph is basically a tree where all degree-1 vertices are on distance 1 from a  central path.

We now proceed with the formal construction.


\begin{description}
\item[Vertices of $N$ and taxa:] Let the taxon set $X = \{x_{i,j}: i \in [r], j \in [n]\}$. For each $i \in [r], j \in [n]$, $N$ will contain a leaf labelled with $x_{i,j}$. The internal vertices of $N$ are $y_{i,j}$ for each $i \in [r], j\in [n]$, and $u_{i,j}, v_{i,j}$ for each $i \in [r], j\in[n]\cup \{0\}$. (Note that some of these vertices will be deleted or suppressed at the end of the construction, in order to turn $N$ into a phylogenetic network with no unlabelled leaves.)

\item[Edges:] The edges of $N$ are as follows.
For each $i \in [r], j \in [n]$, let $\{y_{i,j}, x_{i,j}\}$ be an edge in  $N$.
In addition let $\{u_{i,j-1}, v_{i,j-1}\}$, $\{v_{i,j-1}, y_{i,j}\}$, $\{y_{i,j}, u_{i,j}\}$, $\{u_{i,j}, v_{i,j}\}$ be ``horizontal" edges in $N$.
For each  $i \in [r-1], j \in [n]\cup \{0\}$, let $\{v_{i,j}, u_{i+1,j}\}$ be a ``vertical" edge in $N$.

Finally, we delete all unlabeled degree-$1$ vertices (namely $u_{1,0}$ and $v_{r,n}$), and then suppress all degree-$2$ vertices
(namely $u_{i,0}$ and $v_{i,n}$ for all $i \in [r]$, as well as $u_{1,j}$ and $v_{r,j}$ for all $j \in [n]\cup \{0\}$\markj{, and the vertices $v_{1,0}$ and $u_{r,n}$}).
\textcolor{black}{Note that this
causes $v_{i,0}$ to be adjacent to $v_{i+1,0}$ for $2 \leq i \leq r-2$, and also $u_{i, n}$ to be adjacent to $u_{i+1,n}$
for $2 \leq i \leq r-2$.} See Figure~\ref{fig:gridExample} for an example when $r = 4, n = 11$.


\begin{figure}
\centering
 \includegraphics{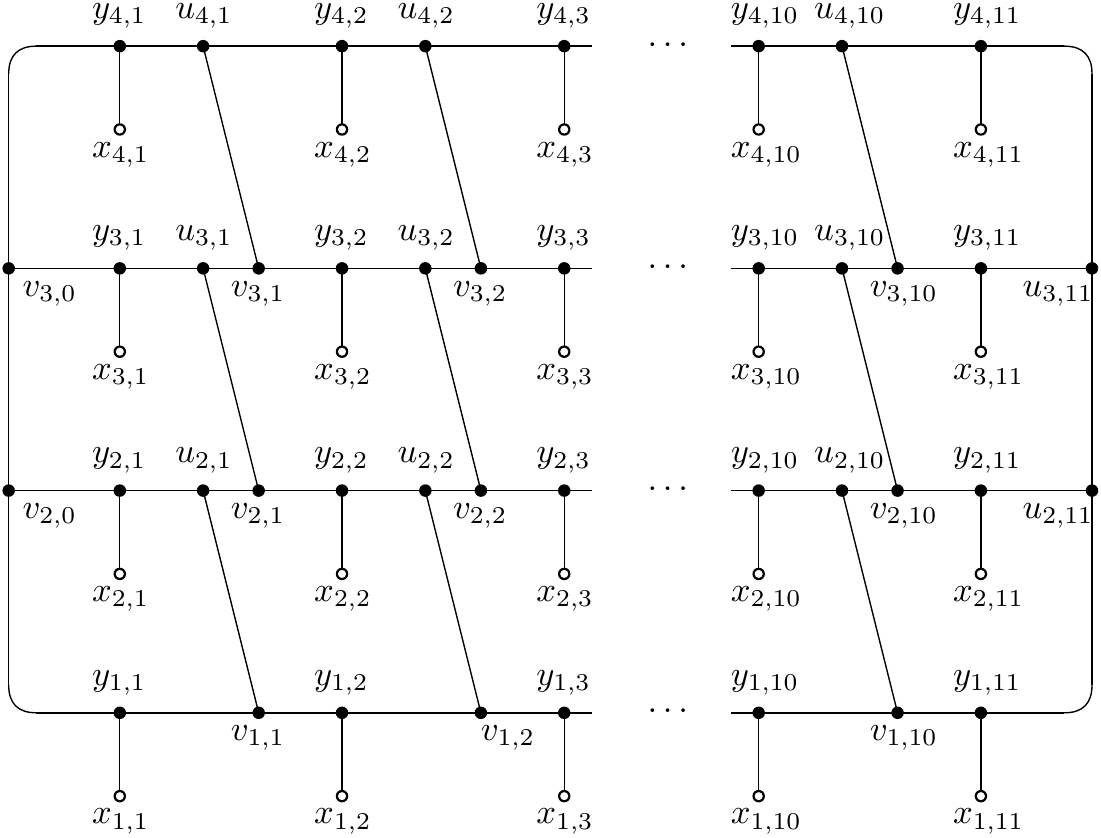}
 \caption{The network $N$ when $r=4$ and $n = 11$.}
 \label{fig:gridExample}
\end{figure}

\item[The tree $T$:] We next construct the tree $T$ as follows. For each $i \in [r], j \in [n]$, $T$ will contain a leaf labelled with $x_{i,j}$.
The internal vertices of $T$ are $z_{i,j}$ for each $i \in [r], j\in [n]$.
For each  $i \in [r], j\in [n]$, there is an edge $\{z_{i,j}, x_{i,j}\}$.
For each $i \in [r]$ and $j\in [n-1]$ there is an edge $\{z_{i,j}, z_{i,j+1}\}$.
Furthermore, for odd $i\in [r-1]$ there is an edge $\{z_{i,n}, z_{i+1,n}\}$,
and for even $r \in [r-1]$ there is an edge $\{z_{i,1}, z_{i+1,1}\}$.
Finally, suppress the degree-$2$ vertices $z_{1,1}$ and $z_{r,1}$ (or $z_{1,1}$ and $z_{r,n}$ when $r$ is odd).
See Figure~\ref{fig:treeExample} for an example when $r = 4, n = 11$.
\end{description}

\begin{figure}
\centering
 \includegraphics{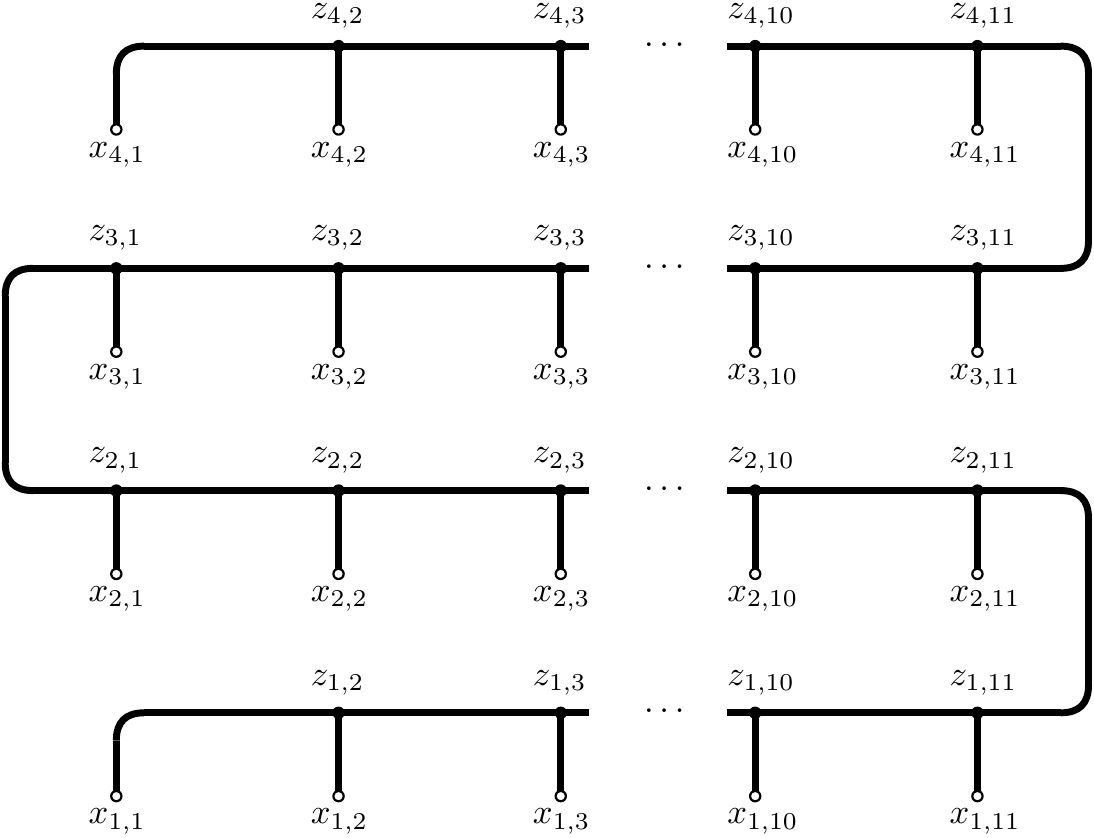}
 \caption{The tree $T$ when $r=4$ and $n = 11$.}
 \label{fig:treeExample}
\end{figure}

\begin{figure}
\centering
 \includegraphics{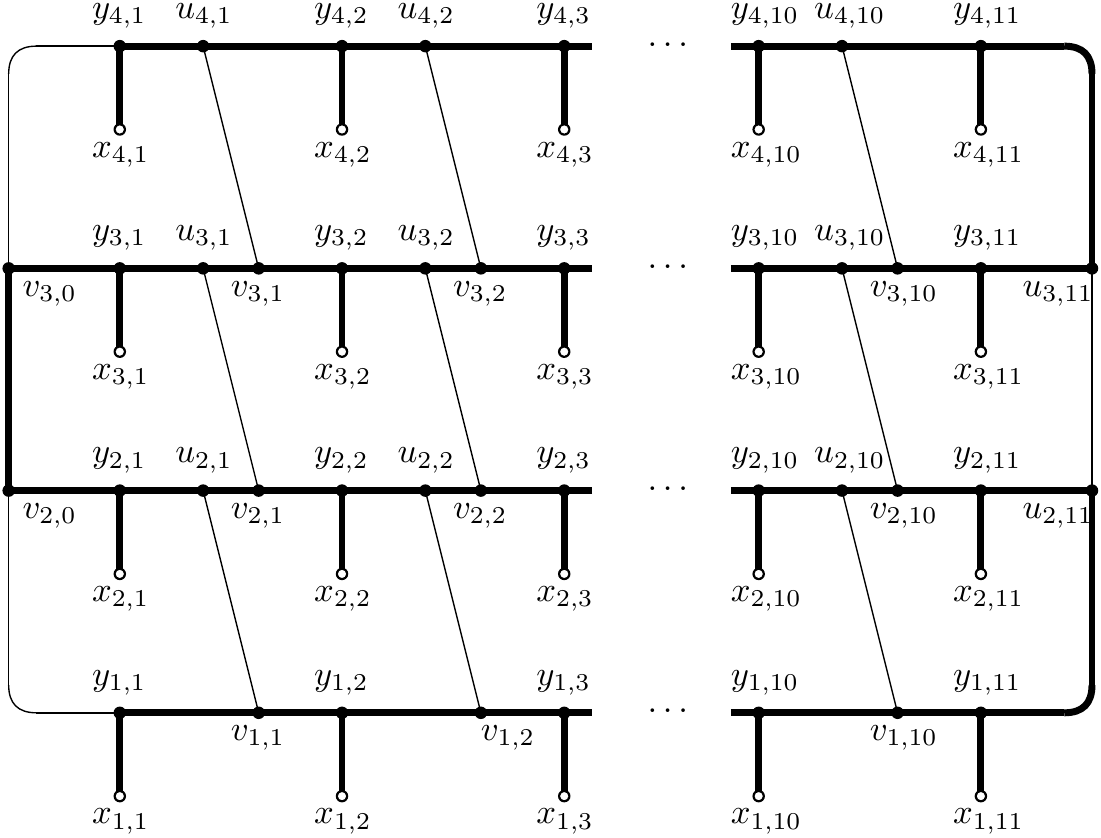}
 \caption{The network $N$ for $r=4, n = 11$, with the tree $T$ drawn in bold.}
 \label{fig:gridExampleTreeEmbedded}
\end{figure}

\begin{lemma}
$T$ is displayed by $N$.
\end{lemma}

\begin{proof}
\markj{Let $N'$ be the network derived from $N$ by deleting edges of the form $\{v_{i,j},u_{i+1,j}\}$, as well as
edges of the form 
$\{u_{i,n},u_{i+1,n}\}$ for $i$ even and $\{v_{i,0},v_{i+1,0}\}$ for $i$ odd,
and the edges $\{x_{1,1}, v_{2,0}\}$, $\{v_{r-1,0}, y_{r,1}\}$. Observe that $N'$ is a subtree of $N$, and that furthermore $N'$ is a subdivision of $T$, which can be seen by mapping internal vertices $z_{i,j}$ of $T$ to $y_{i,j}$. See Figure~\ref{fig:gridExampleTreeEmbedded}.}



\end{proof}

This completes the construction of $N$ and $T$. The display graph $D(N,T)$ is shown in Figure~\ref{fig:displayGraphExample}.
For convenience, we keep the same names for internal vertices of $N$ and $T$ but it will always be clear from the context which structure we are referring to. Note that after suppressing the vertices $x_{i,j}$, vertices $y_{i,j}$ and $z_{i,j}$ are adjacent in $D(N,T)$.

\begin{lemma}
 The treewidth of $N$, $tw(N)$, is equal to $r$.
\end{lemma}
\begin{proof}
To prove that $tw(N)\leq r$, we give a tree decomposition of $N$. We first ignore the nodes $x_{i,j}$ because those can be added to any tree decomposition of the remaining graph by adding the bags $\{x_{i,j},y_{i,j}\}$ and connecting them to any bag containing $y_{i,j}$ for all $i,j$.

We will now give the tree decomposition (in fact path decomposition\footnote{A path-decomposition is a tree decomposition  in which the underlying tree of the decomposition is a path graph.}) of the remaining graph.

Start with the bag $$\{y_{1,1},v_{2,0},\ldots,v_{r-1,0},y_{r,1}\},$$ which contains exactly $r$ nodes. We now sequentially add one node and delete another to get the path decomposition of the remaining graph. Denote the step of adding node $a$ and then deleting node $d$ by the tuple $(a,d)$. \twu{Note that adding node $a$ results in a bag with $r+1$ nodes while deleting nodes $d$ resultes in another bag with $r$ nodes Then} the following steps bring us to the bag $\{v_{i,1}\}_{i\in[r-1]}\cup\{u_{r,1}\}$:
\[(v_{1,1},y_{1,1}), (y_{2,1},v_{2,0}), (u_{2,1},y_{2,1}), (v_{2,1},u_{2,1}), (y_{3,1},v_{3,0}),\ldots, (v_{r-1,1},u_{r-1,1}),(u_{r,1},y_{r,1}). \]

Now we use a similar sequence of steps to go from the bag $\{v_{i,j}\}_{i\in[r-1]}\cup\{u_{r,j}\}$ to the next $\{v_{i,j+1}\}_{i\in[r-1]}\cup\{u_{r,j+1}\}$:
\[
\begin{aligned}
(y_{1,j+1},v_{1,j}), (v_{1,j+1},y_{1,j+1}), &(y_{2,j+1},v_{2,j}), &(u_{2,j+1},y_{2,j+1}), &(v_{2,j+1},u_{2,j+1}),    \\
                                            &(y_{3,j+1},v_{3,j}), &(u_{3,j+1},y_{3,j+1}), &(v_{3,j+1},u_{3,j+1}),    \\
                                            &                     &\ldots                 &                          \\
                                            &(y_{r-1,j+1},v_{r-1,j}), &(u_{r-1,j+1},y_{r-1,j+1}), &(v_{r-1,j+1},u_{r-1,j+1}),\\
(y_{r,j+1},u_{r,j}),(u_{r,j+1},y_{r,j+1}).
\end{aligned}
\]


Finally, do the following sequence of additions and deletions to the bags starting from $\{v_{i,n-1}\}_{i\in [r-1]}\cup\{u_{r,n-1}\}$:
\[(y_{1,n},v_{1,n-1}),(y_{2,n},v_{2,n-1}),(u_{2,n},y_{2,n}),(y_{3,n},v_{3,n-1}),\cdots,(u_{r-1,n},y_{r-1,n}),(y_{r,n},u_{r,n-1}).\]

Hence we get a path decomposition of $N$ minus the nodes $x_{i,j}$ and their incoming edges. This can be seen by inspecting when nodes are added and deleted. Nodes in the initial bag only get deleted, nodes in the final bag only get added, and all other nodes are first added then deleted, therefore we have the running intersection property. It is also clear that each node is in at least one bag, so we still have to check that each edge is represented in a bag. We consider each type of edge separately, and find a bag where the edge is represented.
\begin{itemize}
\item The edges $\{y_{1,1},v_{2,0}\}$, $\{v_{2,0},v_{3,0}\},\ldots,\{v_{r-2,0},v_{r-1,0}\}$ and $\{v_{r-1,0},y_{r,1}\}$ are in the initial bag;
\item The edges $\{v_{i,0},y_{i,1}\}$ for $i\in\{2,\cdots,r-1\}$ are in the intermediate bag for the addition/deletion $(y_{i,1},v_{i,0})$ in the first part of the sequence;
\item $\{u_{i,j},v_{i,j}\}$ for each $i\in\{2,\cdots, r-1\}$ and $j\in [n-1]$ is in the intermediate bag for the addition/deletion $(v_{i,j},u_{i,j})$;
\item $\{v_{i,j},y_{i,j+1}\}$ for each $i\in[r-1]$ and $j\in [n-1]$ is in the intermediate bag for the addition/deletion $(y_{i,j+1},v_{i,j})$;
\item $\{y_{i,j},u_{i,j}\}$ for each $i\in\{2,\cdots, r\}$ and $j\in [n-1]$ is in the intermediate bag for the addition/deletion $(u_{i,j},y_{i,j})$;
\item $\{y_{1,j},v_{1,j}\}$ for each $j\in [n-1]$ is in the intermediate bag for the addition/deletion $(v_{1,j},y_{1,j})$;
\item $\{u_{r,j},y_{r,j+1}\}$ for each $j\in [n-1]$ is in the intermediate bag for the addition/deletion $(y_{r,j+1},u_{r,j})$;
\item $\{v_{i,j},u_{i+1,j}\}$ for each $i\in[r-1]$ and $j\in [n-1]$ is in the intermediate bag for the addition/deletion $(u_{i+1,j},y_{i+1,j})$, this is clear when we realize that $v_{i,j}$ is added in the addition/deletion step $(v_{i,j},u_{i,j})$ or $(v_{1,j},y_{1,j})$ two steps before $(u_{i+1,j},y_{i+1,j})$;
\item The edges $\{y_{i,n},u_{i,n}\}$ for $i\in\{2,\cdots,r-1\}$ are in the intermediate bag for the addition/deletion $(u_{i,n},y_{i,n})$ in the last part of the sequence;
\item The edges $\{y_{1,n},v_{2,0}\}$, $\{u_{2,n},u_{3,n}\},\ldots,\{u_{r-2,n},u_{r-1,n}\}$ and $\{u_{r-1,n},y_{r,n}\}$ are in the final bag.
\end{itemize}
Hence our proposed tree decomposition is indeed a tree decomposition, and the treewidth of $N$ is at most $r$.


%
\textcolor{black}{For the lower bound, observe that the $r \times (n+1)$ grid is a minor of $N$. This grid has treewidth $r$, so $tw(N) \geq r$.} Combining the upper and lower bound, we conclude that the treewidth of $N$ is exactly $r$.
\end{proof}

\begin{figure}
\centering
 \includegraphics{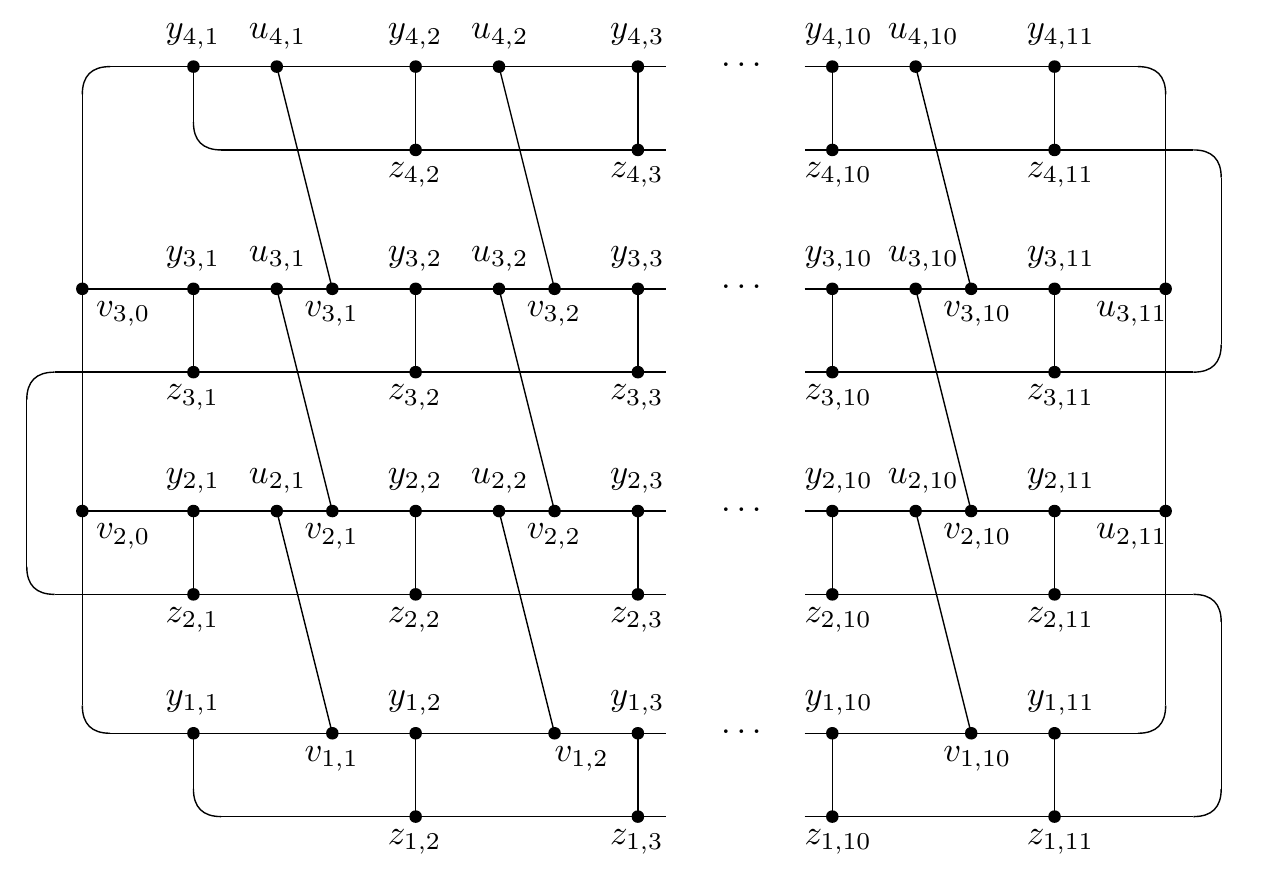}
 \caption{The display graph $D(N,T)$.}
 \label{fig:displayGraphExample}
\end{figure}

In order to  show that $tw(D(N,T))\geq 2r$, we use the concept of \emph{brambles}.
We will construct a bramble in $D(N,T)$ of order $2r+1$. This implies that $tw(D(N,T)) \geq 2r$. The bramble ${\cal B}$ contains the subgraphs induced by $D(N,T)$ on the following sets:

  \begin{itemize}
   \item For each $i \in [r-1]$ and $1 \leq j < n$, the set
  $$S_{i,j} = \{u_{i,l},v_{i,l}, y_{i,l} : l \in [n-1]\cup \{0\}\} \cup \{y_{h,j}, u_{h,j}, v_{h,j}: h \in [r]\}$$

  \item For each $i \in [r]$ and $1 \leq j < n$, the set
 $$T_{i,j} = \{z_{i,l} : l \in [n]\} \cup \{ y_{h,j}, u_{h,j}, v_{h,j}: h \in [r]\}$$

 \item The set $End = \{y_{h,n}, u_{h,n}: h \in [r]\}$

 \item The set $Top = \{y_{r,l}, u_{r,l}: l \in [n-1]\}$

  \end{itemize}

\markj{We note that some of these sets contain vertices such as $v_{1,0}$ that were deleted or suppressed in the construction of $N$. Such vertices should be ignored for the purposes of defining an induced subgraph.}
Intuitively, one may think of the graph $D(N,T)$ as being split up into ``rows" and ``columns", with a ``column" being made up of the vertices $y_{i,j}, u_{i,j},v_{i,j}$ for some fixed $j$ and all values of $i$. A ``row" either consists of all  $y_{i,j}, u_{i,j},v_{i,j}$ for a fixed $i$, or all $z_{i,j}$ for a fixed $i$.
The set $End$ consists of all vertices in the last column, and the set $Top$ consists of all vertices in the top row (except for those already in $End$).
The sets $S_{i,j}$ and $T_{i,j}$ combine all vertices from a given row and column (except those vertices already in $End$).
Note that $End$ is vertex-disjoint from all the other sets; this will be crucial for the lower bound on the order of ${\cal B}$.

  \begin{lemma}\label{lem:bramble}
   ${\cal B}$ is a bramble in $D(N,T)$.
  \end{lemma}
  \begin{proof}
   Observe that all the sets induce a connected subgraph of $D(N,T)$.
   (In particular, the ``columns" are connected because of the edges $\{v_{i,j},u_{i+1,j}\}$; also note that for $T_{i,j}$ the sets $\{z_{i,l} : l \in [n-1]\}$ and $\{ y_{h,j}, u_{h,j}, v_{h,j}: h \in [r]\}$ are connected by the edge $\{z_{i,j},y_{i,j}\}$.)
   It remains to show  that for each pair of sets in ${\cal B}$ the sets either share a vertex or are joined by an edge with one vertex in each set.

   To see that the sets $Top$ and $End$ touch, observe that $Top$ contains $u_{r,n-1}$ and $End$ contains $y_{r,n}$, and these vertices are connected by an edge.
   To see that $End$ touches the other sets, observe that all other sets contain either the vertex $z_{i,n}$ or $v_{i,n-1}$ for some $i \in [r]$.
  As both of these vertices are adjacent to $y_{i,n}$, it follows that $End$ is touches each of these sets.

  To see that $Top$ touches each of the other sets except for $End$, observe that each of these sets contains $y_{r,j}$ for some $1 \leq j < n$. As $y_{r,j}$ is also in $Top$, these sets touch.

  It remains to consider pairs of sets where each set is $S_{i,j}$ or $T_{i,j}$ for some
 $i \in [r]$ and $j \in [n-1]$.
First consider a set $S_{i,j}$ and a set $T_{i',j'}$.
As both these sets contain $y_{i,j'}$, the sets touch.
Next consider sets $S_{i,j}$ and  $S_{i',j'}$.
As both these sets contain $y_{i,j'}$, the sets touch.
Finally consider the set $T_{i,j}$ and  $T_{i',j'}$.
Then $T_{i,j}$ contains $z_{i,j'}$ and $T_{i',j'}$ contains $y_{i,j'}$. As these vertices are adjacent, the sets touch.
\end{proof}

\begin{lemma}\label{lem:brambleOrder}
The order of ${\cal B}$ is  $2r + 1$.
\end{lemma}

\begin{proof} Observe that the \twu{ set $\{y_{i,2}, z_{i,2}: i \in [r]\} \cup \{y_{1,n}\}$ is} a hitting set of size $2r+1$.

To see that any hitting set must have size at least $2 r +1$,
suppose for a contradiction that $H$ is a hitting set for $\cal B$ with $|H| \leq 2r$.
As $n >  2r + 2$, there exists some $1 < j < n$ such that $H$ does not contain $u_{i,j},v_{i,j},y_{i,j}$ or $z_{i,j}$ for any $i \in [r]$.
For each $i \in [r]$, $H$ contains elements from $T_{i,j}$,
from which it follows that $H$ must contain some element from  $\{z_{i,l}: l \in [n]\}$  for each $i \in [r]$.
Similarly as $H$ contains elements from $S_{i,j}$
,  $H$ must contain some element from  $\{u_{i,l},v_{i,l}, y_{i,l} : l \in [n-1] \cup\{0\}\}$ for each $i \in [r-1]$.
In addition, $H$ must contain some element from $Top = \{y_{r,l}, u_{r,l}: l \in [n-1]\}$.

As these sets are 
disjoint and there are $2r$ of them, $H$ must contain exactly one element from each of these sets.
But as each of these sets is disjoint from  $End = \{y_{h,n}, u_{h,n}: h \in [r]\}$, it follows that $H$ contains no element of $End$, a contradiction.
\end{proof}

This shows that the treewidth of the display graph $D(N,T)$ is at least $2r$. From the above three lemmas we have the following:


\begin{theorem}
For any positive integer $r$, there is a network $N$ of treewidth $r$ and a tree $T$ such that $N$ displays $T$ and $tw(D(N,T)) \geq 2r$.
\end{theorem}

\section{Discussion and conclusions}
\label{sec:discussion}


An obvious open question is whether we can match the theoretical upper and  constructive lower bound on the treewidth of $D(N,T)$ in terms of the treewidth of $N$. This means either finding a tight example of the inequality $tw(D(N,T)) \leq 2tw(N)+1$, \textit{or} improving the upper bound to match the $2tw(N)$ lower bound of the construction from the previous section. It is also natural to explore \emph{empirically} how large the treewidth of $D(N,T)$ is compared to the treewidth of $N$, when $N$ displays $T$. We conjecture that for realistic phylogenetic trees and networks $tw(D(N,T))$ will be much smaller than $2tw(N)$.

As touched upon
in Section \ref{sec:networks} it could additionally be interesting to identify non-trivial examples when $N$ does not display $T$ but $tw(D(N,T)) = tw(N)$ and to give, if possible, a phylogenetic interpretation to this. \sfinal{Phylogenetics has defined many topologically-restricted subclasses of phylogenetic networks, such as \emph{tree-based} networks \cite{francis2018tree}, precisely to prohibit networks (such as that shown in Figure \ref{fig:net}) that are artificially large and complex with respect to the number/location of taxa in the network. Possibly the display relation will behave differently on such restricted subclasses with respect to $tw(D(N,T))$.} In any case, recent advances in treewidth solvers will be useful here (see e.g. \cite{berndt2018}) since display graphs can quickly become quite large. We now understand that, after suppression of degree-2 nodes, display graphs \sfinal{of two phylogenetic trees} are exactly those (biconnected, cubic) graphs of tree arboricity 2; is there any hope of computing treewidth quickly on these graphs? See the related discussion in \cite{DBLP:journals/tcs/KelkSW18}.

Algorithmically, the obvious challenge that (still!) remains is to convert MSOL formulations into practical dynamic programming algorithms running over tree decompositions. This remains  tempting, for the following reason. In \cite{kelk2015} it is reported that display graphs of two trees $T_1$, $T_2$ often have low treewidth compared to even conservative phylogenetic dissimilarity measures on $T_1, T_2$, such as Tree Bisection and Reconnect (TBR) distance, and this makes computation of these measures (paramerized by treewidth of the display graph) attractive. But what about networks - as opposed to display graphs? In phylogenetics it is quite common to construct phylogenetic networks by asking for a network $N$ that simultaneously displays two (or more) trees $T_1, T_2$ and which minimizes $r(N)$; this is the well-studied \emph{hybridization number} problem \cite{sempbordfpt2007,vanIersel20161075}.  In such an $N$, $r(N)$ will be  equal to the TBR-distance of $T_1$ and $T_2$ \cite{van2017unrooted} which, as mentioned earlier, can be large compared to $tw( D(T_1, T_2) )$. The question arises how $tw(N)$ relates to $tw( D(T_1, T_2) )$ and, in particular, whether it is also ``low''. If so, there is some hope that phylogenetic networks arising in practice will also have low treewidth, compared to other phylogenetic measures. More empirical study is needed in this area.

The obvious theoretical shortcoming of this approach is that phylogenetic MSOL formulations are complex and explicit dynamic programs require some effort to write and understand (see e.g. \cite{baste2017efficient}) with relatively high exponential dependency on the treewidth \sfinal{bound}. The UTC formulation in this article nevertheless seems a promising candidate for a ``clean'' explicit dynamic program since it has, by phylogenetic standards, a comparatively straightforward combinatorial structure.

Looking forward we observe that, as phylogenetic networks become more commonplace in computational biology, it is natural to compare networks, rather than trees (see e.g. \cite{francis2018bounds,janssen2018exploring}). In this regard, \emph{network-network} display graphs are certainly worthy of investigation. For example, it is straightforward to prove that if two phylogenetic networks $N_a, N_b$ both display a tree $T$, $tw(D(N_a, N_b)) \leq r(N_a) + r(N_b) + 2$. Now, if $N_a$ and $N_b$ are two distinct optima (i.e. competing hypotheses) produced by an algorithm solving the hybridization number problem for two trees $T_1, T_2$, then $r(N_a)$ and $r(N_b)$ are both equal to the TBR-distance $d$ of $T_1$ and $T_2$ \cite{van2017unrooted}. Hence, $tw(D(N_a, N_b)) \leq 2d + 2$. In particular: the treewidth of the display graph formed from the networks, will be bounded as a function of the TBR-distance of the two original trees.
\snew{Similarly, the proof of Lemma \ref{lem:twbound} goes through essentially unchanged for two networks on the same set of taxa: if $N_2$ displays $N_1$ then $tw(D(N_2, N_1)) \leq 2tw(N_2) + 1$.}

Perhaps such treewidth bounds can help in the development of compact FPT MSOL proofs for determining the dissimilarity of networks. There is quite some potential here. Topological decompositions in phylogenetics (into quartets, triplets, \emph{agreement forests} and so on) can be modelled fairly naturally within MSOL \cite{kelk2015}. Higher-order analogues are emerging for decomposing phylogenetic networks (see e.g.  \cite{huber2017reconstructing}) - and it is plausible that such structures could also be encoded within MSOL.

Finally, stepping away from phylogenetics, the study of display graphs continues to generate interesting new questions for algorithmic graph theory. In particular, the behaviour (and ``phylogenetic meaning'') of (forbidden) minors in display graphs remains a subject where much is still to be learned \cite{fernandez2018compatibility,DBLP:journals/tcs/KelkSW18}. Indeed,
display graphs can be viewed as a special case of a more generic problem. Given a set of graphs and a well-defined protocol for merging them, how do parameters of the merged graph (and topological features such as minors) relate to parameters and features of the constituent graphs?




\section*{Acknowledgements}
Mark Jones and Remie Janssen were supported by Leo van Iersel's Vidi grant (NWO): 639.072.602. Georgios Stamoulis was supported by an NWO TOP 2 grant. Part of the work was supported by CNRS ``Projet international de cooperation
scientifique (PICS)'' grant number 230310 (CoCoAlSeq).

\appendix
\section{Appendix}

\subsection{Unrooted tree compatibility (UTC) is FPT when parameterized by treewidth: a proof via Courcelle's Theorem}
\label{subsec:whatever}
\noindent
This leverages the upper bound on $tw(D(N,T))$ as a function of the treewidth $tw(N)$ of $N$ proven earlier in the paper, see Lemma \ref{lem:twbound}.

The high-level idea of the following MSOL formulation is that, if $N$ displays $T$, then (as discussed in Section \ref{sec:networks}) $N$ contains some subtree $T'$
that is a subdivision of $T$ and which can be ``grown'' into a spanning tree $T''$ of $N$.
Spanning trees of $N$ are precisely those subgraphs obtained by deleting a subset of edges $E'$ from $N$ to make it connected and acyclic. Note that the set of quartets (unrooted phylogenetic trees on subsets of exactly
4 taxa) displayed by $T''$ is identical to those displayed by $T'$, which is identical to those displayed by $T$. (In other words,
subdivision operations, and pendant subtrees without taxa that possibly hang from $T''$, do not induce any extra quartets.)

The core idea underpinning MSOL is to query properties of a graph using universal and existential quantification ranging not just over vertices and edges, but also subsets of these objects. For the benefit of readers not familiar with MSOL we now show how various basic auxiliary predicates can be easily constructed and combined to obtain more powerful predicates.
(The article \cite{kelk2015} gives a more comprehensive inroduction to the use of these techniques in phylogenetics). The MSOL sentence will be queried over the display graph
$D(N,T)$ where we let $V$ be the vertex set of $D(N,T)$ and $E$ its edge set. Here $R^{D}$ is the edge-vertex incidence relation on $D(N,T)$. We let $V_T, V_N, E_T, E_N$ denote
those vertices and edges of $D(N,T)$ which belong to $T, N$ respectively (note that \sfinal{$V_T \cap V_N = X$}). Alongside  $X, V, E$ all this information is available to the MSOL
formulation via its \emph{structure}.

\begin{itemize}
\item test that $Z$ is equal to the union of two sets $P$ and $Q$:
\begin{flalign*}
P \cup Q = Z  := &\forall z ( z \in Z \Rightarrow z \in P \vee z \in Q)&\\
&\wedge \forall z ( z \in P \Rightarrow z \in Z) \wedge \forall z ( z \in Q \Rightarrow z \in Z).
\end{flalign*}

\item  test that $P \cap Q = \emptyset$:
\begin{flalign*}
\mathrm{NoIntersect}(P,Q) := &\forall u \in P( u \not \in Q ).&
\end{flalign*}

\item test that $P \cap Q = \{v\}$:
\begin{flalign*}
\mathrm{Intersect}(P,Q,v) := &(v \in P) \wedge (v \in Q) \wedge \forall u \in P( u \in Q \Rightarrow (u  = v) ).&
\end{flalign*}

\item test if the sets $P$ and $Q$ are a bipartition of $Z$:
\begin{flalign*}
\mathrm{Bipartition}(Z, P, Q) := &(P \cup Q = Z) \wedge \mathrm{NoIntersect}(P,Q).&
\end{flalign*}

\item test if the elements in $\{x_1, x_2, x_3, x_4\}$ are pairwise different:
\begin{flalign*}
\mathrm{allDiff}( x_1, x_2, x_3, x_4 ) := &\bigwedge_{i \neq j \in \{1,2,3,4\}} x_i \neq x_j.&
\end{flalign*}

\item check if the nodes $p$ and $q$ are adjacent:
\begin{flalign*}
\mathrm{adj}(p,q) :=& \exists e \in E ( R^{D}(e,p) \wedge R^{D}(e,q)).&
\end{flalign*}
\end{itemize}

The complex predicate $PAC(Z, x_1, x_2, K)$  \sfinal{(``path avoiding edge cuts?'')} asks: is there a path from $x_1$ to $x_2$ entirely contained inside vertices $Z$ that avoids all the edges $K$? We model this by observing
that this does \emph{not} hold if you can partition $Z$ into two pieces $P$ and $Q$, with $x_1 \in P$ and $x_2 \in Q$, such that the only edges that cross the induced cut (if
any) are in $K$.
\begin{flalign*}
PA&C(Z, x_1, x_2, K) :=\\
&(x_1 = x_2) \vee \neg \exists P, Q \Bigg( \mathrm{Bipartition}(Z,P,Q) \wedge x_1 \in P \wedge x_2 \in Q \\
&\wedge\Big(\forall p, q \Big(p \in P \wedge q \in Q \Rightarrow \neg \mathrm{adj}( p,q) \vee \left(\exists g \in K\left( R^{D}(g, p) \wedge R^{D}(g,q) \right)\right)\Big)\Big)\Bigg)
\end{flalign*}

The following predicate $QAC^{i}$ \sfinal{(``quartet avoiding edge cuts?'')}, where $i \in \{T,N\}$, returns true if and only if $i$ contains an \sfinal{image} of quartet $x_a x_b | x_c x_d$ that is disjoint from the edge cuts $K$. \sfinal{As usual we write $x_a x_b | x_c x_d$ to denote the quartet where the path between $x_a$ and $x_b$ is disjoint from the path between $x_c$ and $x_d$. (The tree $T$ shown in Figure \ref{fig:net}, for example, is the quartet $ab|cd$)}.
\begin{flalign*}
QA&C^{i}(x_a, x_b, x_c,x_d, K) :=\\
&\exists u, v \in V_i \Bigg(\mathrlap{ (u \neq v) \wedge \exists A,B,C,D,P \subseteq V_i \Big( u \in P \wedge v \in P  }\\
& \wedge x_a, u \in A &&\wedge x_b, u \in B &&\wedge x_c, v \in C\\
& \wedge x_d,v \in D && \wedge  \mathrm{Intersect}(A,B,u) &&\wedge \mathrm{Intersect}(A,P,u)\\
& \wedge \mathrm{Intersect}(B,P,u) &&\wedge \mathrm{Intersect}(C,D,v) &&  \wedge \mathrm{Intersect}(C,P,v)\\
& \wedge \mathrm{Intersect}(D,P,v) &&\wedge   \mathrm{NoIntersect}(A,C) &&\wedge \mathrm{NoIntersect}(B,C)\\
& \wedge \mathrm{NoIntersect}(A,D) &&\wedge \mathrm{NoIntersect}(B,D)   &&\wedge PAC(A, u, x_a, K) \\
&\wedge PAC(B, u, x_b, K) &&\wedge PAC(C, v, x_c, K) &&\wedge PAC(D, v, x_d, K)& \\
&&&&&\wedge PAC(P, u, v, K)\Big)\Bigg)&
\end{flalign*}

We need a prediate which asks: is the subgraph induced by vertex subset $Z$, and then with edges $K$ deleted, connected? We model
this as follows: for every pair of vertices $u$ and $v$ in $Z$ a path should exist from $u$ to $v$ completely contained inside
$Z$ and which avoids the edges $K$. Hence,
\begin{eqnarray*}
\mathrm{Connected}(Z, K) & := \forall u, v \in Z ( PAC( Z, u, v, K )).
\end{eqnarray*}

In a similar vein, we need a predicate which asks: is the subgraph induced by vertex subset $Z$, and then with edges
$K$ deleted, \emph{acyclic?}  The idea here is that, if it is not acyclic, there will exist two distinct vertices $u, v \in Z$
such that $u$ can reach $v$ via two distinct, vertex-disjoint paths $P$ and $Q$:
\begin{eqnarray*}
\mathrm{Acyclic}(Z, K) &:=& \neg \exists u, v \in Z\Big(\exists P, Q \subseteq Z\big( u\neq v \wedge  P \cap Q = \{u,v\}  \\
&& \wedge P \neq Q \wedge PAC( P, u, v, K ) \wedge PAC( Q, u, v, K )\big)\Big).
\end{eqnarray*}
\noindent
(The predicate $P \cap Q = \{u,v\}$ is a simple modification of the earlier $\mathrm{Intersect}$ predicate.)

The final formulation is shown as below. The first line asks for a subset $E'$ (representing the edges we delete \sfinal{from $N$} to obtain $T''$), the second line requires that the $N$ part of $D(N,T)$ remains
connected and acyclic after deletion of $E'$ (and thus induces a spanning tree), and from the third line onwards we stipulate that, after deletion of $E'$, the set of quartets that survive
is exactly the same as the set of quartets displayed by $T$. (This is leveraging the well-known
result from phylogenetics that two trees are compatible if and only if they display the same
set of quartets \cite{SempleSteel2003}).  Note that the overall length of the MSOL fragment is fixed i.e. it is not dependent on parameters of the input.
\begin{flalign*}
\exists E' &\mathrlap{\subseteq E_N \Bigg( \mathrm{Connected}(V_N, E') \wedge \mathrm{Acyclic}(V_N, E') }\\
&\mathrlap{\wedge \forall x_1, x_2, x_3, x_4 \in X\Bigg( \mathrm{allDiff}(x_1, x_2, x_3, x_4)  }\\
& \Rightarrow & \Big(&\big(QAC^{T}(x_1, x_2, x_3, x_4, \emptyset)    && \Leftrightarrow QAC^{N}(x_1, x_2, x_3, x_4, E') \big)& \\
&&&  \wedge \big(QAC^{T}(x_1, x_3, x_2, x_4, \emptyset) && \Leftrightarrow QAC^{N}(x_1, x_3, x_2, x_4, E') \big)& \\
&&&  \wedge \big(QAC^{T}(x_1, x_4, x_2, x_3, \emptyset) && \Leftrightarrow QAC^{N}(x_1, x_4, x_2, x_3, E') \big)\Big)&\Bigg)\Bigg).\\
\end{flalign*}

\subsection{MSOL proof for recognizing display graphs}
\label{subsec:msolrecog}


The following MSOL fragment checks whether a cubic, simple graph $G=(V,E)$ is a suppressed display graph. We re-use predicates defined
in the previous section.
\[
\exists V_1, V_2( \mathrm{Bipartition}(V, V_1, V_2) \wedge_{i=1,2} \mathrm{Connected}(V_i, \emptyset)
\wedge_{i=1,2} \mathrm{Acyclic}(V_i, \emptyset) ).
\]

\bibliographystyle{abbrvurl}
\bibliography{JGAA_final}

\end{document}